\documentclass[11pt,english]{article}
\usepackage[utf8]{inputenc}
\usepackage[a4paper, total={6in, 9in}]{geometry}
\usepackage[pdftex]{graphicx}
\usepackage[table]{xcolor}
\usepackage{float}
\usepackage{hhline}
\usepackage{amssymb}
\usepackage{amsmath}
\usepackage{amsthm}
\usepackage{eurosym}
\usepackage{hyperref}
\usepackage{hhline}
\usepackage{soul}
\newtheorem{definition}{Definition}
\newtheorem{proposition}{Proposition}
\newtheorem{theorem}{Theorem}
\newtheorem{example}{Example}
\newtheorem{lemma}{Lemma}
\newtheorem{corollary}{Corollary}
\newcommand{\commentout}[1]{}

\title{The Value of Mediation in Long Cheap Talk}
\author{Itai Arieli \and Ivan Geffner \and Moshe Tennenholtz\thanks{The work by Ivan Geffner and Moshe Tennenholtz was supported by funding from
the European Research Council (ERC) under the European Union’s Horizon 2020
research and innovation programme (grant agreement 740435).}}
\date{}

\begin{document}

\maketitle

\begin{abstract}
In this paper, we study an extension of the classic long cheap talk equilibrium introduced by Aumann and Hart~\cite{aumann-hart-03}, and ask how much can the players benefit from having a trusted mediator compared with the standard unmediated model. 
We focus on a setting where a fully informed sender without commitment power must disclose its information to influence the behavior of a self-interested receiver. 
We show that, in the case of binary actions, even though a mediator does not help neither the sender nor the receiver directly, it may still allow improving the payoff of an external decision-maker whose utility is affected by the realized state and the receiver's action. 
Moreover, we show that if there are more than two actions, there exist games in which both the sender and the receiver simultaneously benefit from mediation.
\end{abstract}

\section{Introduction}

The study of Strategic Information Transmission in economic theory has been ongoing for many years. 
In these models, there are usually two phases: a communication phase in which the players communicate with each other, and an action phase where the players  play an action based on the outcome of the communication phase. In this work, we study the influence of a trusted third-party mediator on the outcome efficiency and on the utility of the players in such models. 

We focus on \emph{information transmission games}, which are games where a fully informed sender must strategically disclose information about her type to the receiver in order to influence the action it takes. The utility of both players depends on the sender's type, which represents the realized state of the world, and on the action played by the receiver. We study two models of communication, the \emph{long cheap talk} model, and the \emph{mediated model}. The long cheap talk model is an extension of the one introduced by Aumann and Hart~\cite{aumann-hart-03} where players freely exchange messages in arbitrarily many rounds of communication and play an action afterwards. The mediated model is analogous, except that players can only communicate with a trusted third party mediator (instead of communicating with other players).

Our aim is to find how much a system benefits from having mediated communication in contrast to long cheap talk. For this purpose, we extend Ashlagi, Monderer and Tennenholtz's ~\cite{ashlagi2008value} 
notion of \emph{value of mediation}\footnote{The term \emph{value of mediation} has been used several times in the literature with different meanings. Even though at a high level they all describe the benefit of having mediated communication as opposed to unmediated communication, the details of each definition differ. See Section~\ref{sec:related-literature} for details.} as follows: given an information transmission game $\Gamma$ and a utility function $w$ that represents the total welfare of the system, we define the \emph{value of mediation} of $\Gamma$ (w.r.t. $w$) as the ratio between the maximal $w$-utility that can be achieved in equilibrium in the mediated model and the maximal $w$-utility that can be achieved in equilibrium in the long cheap talk model.
Intuitively, a higher value of mediation implies that the system benefits more from mediated communication relative to long cheap talk. Moreover, it is straightforward to check that having a mediator never hurts the system, and thus that the value of mediation is at least $1$ for all games $\Gamma$ and all positive welfare functions $w$.

Our first result shows that, in all information transmission games where the receiver has only two actions, the value of mediation is $1$ for all welfare functions that are monotone over the utilities of the sender and the receiver. In particular, this shows that a mediator cannot increase the utility of the sender or that of the receiver in equilibrium. However, we show that there exist games with two actions and welfare functions $w$ such that the value of mediation is arbitrarily high. This means that, in the case of binary actions, even though a mediator cannot help the sender or the receiver directly, it can still be used to greatly increase the utility of a third-party agent.

Our second result shows that, in contrast to the case with only two actions, there exist games with three actions in which the value of mediation is strictly positive for both the utility of the sender and the utility of the receiver. This shows that the restriction to binary actions in our first result is tight.

\subsection{Related Literature}\label{sec:related-literature}

The study of strategic information transmission started with the seminal work of Crawford and Sobel~\cite{crawford1982strategic}, who studied information transmission games in which the sender and the receiver only have one round of communication. The concept of \emph{long cheap talk}, where players can communicate using arbitrarily many rounds, was long present in the distributed computing and cryptography literature, especially in that related to the design and analysis of multiparty protocols that are resilient to faulty behavior~\cite{pease1980reaching, ben1983another, rabin1983randomized, bracha1985asynchronous, yao1986generate, bgw88, ben1993asynchronous}. 
However, it was not until some time later that long cheap talk was also studied from a game theoretic point of view~\cite{forges1990equilibria,
aumann-hart-03,
krishna2004art,
ADGH06,
goltsman2009mediation,
adgh19}. Forges~\cite{forges1990equilibria} analyzed the effects of long cheap talk in a setting that emulates a job interview. She also pointed out the difference between the outcome of long cheap talk versus (one round) cheap talk. 
Ben-Porath~\cite{ben1998correlation} gave a way to simulate mediated communication using long cheap talk if players have a punishment strategy. This work was later generalized by Abraham et al.~\cite{ADGH06, abraham2008lower, adgh19}, and Geffner and Halpern~\cite{geffner2023lower,geffner2023communication}, who gave a way to simulate mediated communication even in the presence of coalitions up to a certain size. A direct consequence of these results is that, if we consider information transmission games with several senders (these games are sometimes referred to as \emph{information aggregation games}~\cite{rausser2015rational,arieli2023resilient}), the value of mediation is $1$ with respect to any welfare function if the number of senders is large enough. In particular, the value of mediation is $1$ if there are at least three senders or if there are at least two senders and players have a punishment strategy. This implies that, even with two senders, a mediator cannot Pareto-improve the utilities of all of the players involved (for a short discussion about these results, see Section~\ref{sec:information-aggregation}). 
Aumann and Hart~\cite{aumann-hart-03} provided a complete characterization of the set of Bayesian Nash equilibria in information transmission games with long cheap talk. However, their characterization is mainly abstract and, as a consequence, it cannot be used to construct an algorithm that tells if mediated communication improves over long cheap talk, even for specific classes of games. Krishna and Morgan~\cite{krishna2004art} showed that the inclusion of additional rounds of communication in the model of Crawford and Sobel can Pareto-improve the utility of both the sender and the receiver with respect to one round cheap talk. Goltsman et al.~\cite{goltsman2009mediation} compared the outcome between mediated communication and long cheap talk in a family of information transmission games where the utility of the sender and the receiver have a specific form. They characterize the games in this family in which long cheap talk performs as well as mediated communication. Our work complements these results by (a) showing that mediated communication does not improve over long cheap talk in \emph{any} information transmission game with binary actions, and (b) by providing an information transmission game with three actions in which mediated communication provides a Pareto-improvement over long cheap talk. Even though an example in which mediation Pareto-improves long cheap talk already exists in the literature (see, for instance, Goltsman et al.~\cite{goltsman2009mediation}), to the best of our knowledge this is the first \emph{finite} example. Moreover, as a consequence of (a), the example we provide is the smallest possible with respect to the number of actions.

Our work is also closely related to the analysis of the efficiency of strategic interactions. The classic measure of efficiency in games is the Price of Anarchy (see Koutsoupias and Papadimitriou~\cite{koutsoupias1999worst}  and Roughgarden and Tardos~\cite{roughgarden2007introduction}). Inspired by this concept, Ashlagi, Monderer and Tennenholtz~\cite{ashlagi2008value} introduce two notions to measure the influence of mediation.
The first one is the \emph{enforcement value}, which measures the ratio between the maximal welfare of a strategy profile (not necessarily in equilibrium) and the maximal welfare obtained in a correlated equilibrium, and the second one is the  the \emph{value of mediation}, which measures the ratio between the maximal welfare obtained in a correlated equilibrium and the maximal welfare obtained in a mixed-strategy equilibrium. They provide a characterization for these measures in some families of two-player games. We expand on their definition of \emph{value of mediation} first by allowing players to have private types, and second by allowing certain flexibility when defining what we consider to be the welfare of the system. Typically, the welfare of the system is assumed to be the sum of the utilities of all the agents involved. However, in some cases we might be more interested on analyzing the utility of the sender (in this case we take the welfare function to simply be the utility of the sender), analyzing that of the receiver, or to analyze a completely different welfare function $w$.
For this purpose, we allow our definition to take the welfare function $w$ as an input and output the ratio between the maximum of $w$ in a mediated equilibrium and the maximum of $w$ in a long cheap talk equilibrium. 
As mentioned in the previous section, there are other notions of \emph{value of mediation} in the literature. For instance, given an information transmission game, Salamanca~\cite{salamanca2021value} defines the value of mediation as the maximum utility that the sender can get in a mediated equilibrium, and Corrao and Dai~\cite{corrao2023mediated} define it as the difference between the maximum utility that the sender can get in a mediated equilibrium and the maximum utility that the sender can get in a (one round) cheap talk equilibrium.

\section{Formal Definitions}\label{sec:formal-defs}

In this section we provide the definitions of the main concepts used throughout the paper.

A \emph{Bayesian game} is a tuple $\Gamma = (P, T, q, A, U)$ in which 
\begin{itemize}
\item $P = \{1,2,\ldots, n\}$ is the set of players.
\item $T = T_1 \times T_2 \times \ldots \times T_n$ is the set of possible type profiles (where $T_i$ is the set of possible types of player $i$).
\item $q \in \Delta T$ is a distribution of type profiles that is common knowledge between the players.
\item $A = A_1 \times A_2 \times \ldots \times A_n$ is the set of possible action profiles that players can play (where $A_i$ is the set of possible actions for player $i$).
\item $U = (u_1, \ldots, u_n)$ is an $n$-tuple of utility functions where $u_i : T \times A \longrightarrow \mathbb{R}$ outputs the utility of player $i$ given the types of the players and the action profile played.
\end{itemize}

A \emph{strategy} in $\Gamma$ for player $i$ is a map $\mu_i :
T_i \rightarrow \Delta(A_i)$ from types to actions. Intuitively, 
a strategy in a Bayesian game tells player $i$ how to choose its
action given its type. Since the distribution $q$ is common knowledge, the expected utility of player $i$ when players play strategy profile $\vec{\mu} = (\mu_1, \mu_2, \ldots, \mu_n)$ is given by 

\begin{equation}\label{eq:expected-payoff}
u_i(\vec{\mu}) = \sum_{t_i \in T_i} q(t_i) \sum_{\vec{t}}
q(\vec{t} \mid t_i) u_i(\vec{t}, \vec{\mu}(\vec{t})),
\end{equation}
where 
$u_i(\vec{t}, \vec{\mu}(\vec{t}))$ denotes the expected utility of player
$i$ when the type profile is $\vec{t}$ and the action profile 
is chosen
according to 
$(\mu_1(t_1), \ldots, \mu_n(t_n))$, and $q(\vec{t} \mid t_i)$ denotes the probability that the type profile is $\vec{t}$ given that player $i$'s type is $t_i$.
This allows us to define Bayesian Nash equilibrium as follows:

\begin{definition}\label{def:bayesian-nash}
In a Bayesian Game $\Gamma = (P,T, q, A, U)$, a strategy profile
$\vec{\mu} := (\mu_1, \ldots, \mu_n)$ 
is a \emph{(Bayesian) Nash equilibrium} if, 
for
all players $i$ and all strategies $u'_i$ for $i$,
$$u_i(\vec{\mu}) \ge u_i(\vec{\mu}_{-i}, \mu'_i).$$
\end{definition}

Intuitively, a strategy profile is a Nash equilibrium if no player can increase its expected payoff by switching to a different strategy.

\subsection{Games with Communication}\label{sec:com-games}

Given a Bayesian game $\Gamma = (P, T, q, A, U)$ (which we call the \emph{underlying game}), consider two extensions $\Gamma_{CT}$ and $\Gamma_d$ in which players can exchange messages before playing an action in $\Gamma$. In the \emph{long cheap talk extension} $\Gamma_{CT}$, players can freely communicate with other players, while in the \emph{mediator extension} $\Gamma_{d}$, players can only do so with a trusted third-party mediator.

Similarly to the model introduced by Aumann and Hart~\cite{aumann-hart-03}, both in $\Gamma_{CT}$ and $\Gamma_d$, players first have a \emph{communication phase} in which they can send and receive messages via private authenticated channels, and then they have an \emph{action phase} in which players must play an action in the underlying game $\Gamma$. The communication is \emph{synchronous} in both extensions. This means that the communication proceeds in \emph{rounds} and all messages sent at round $r$ are guaranteed to be received by the recipient at the beginning of round $r+1$. Intuitively, a round of communication $r$ for player $i$ should go as follows:
\begin{enumerate}
\item Player $i$ receives all messages sent to her at round $r-1$.
\item $i$ computes, for each player $j$, which messages $m_1^j, m_2^j, \ldots$ it should send to $j$.
\item $i$ sends $m_1^j, m_2^j, \ldots$ to $j$ (which will be received by $j$ at the beginning of round $r+1$).
\end{enumerate}

The local history $h_i$ of player $i$ consists of $i$'s type and the sequence of all messages sent and received by $i$ at each round, along with all the internal computations (including possible randomization) performed by player $i$. A strategy $\sigma_i$ for player $i$ in $\Gamma_{CT}$ or $\Gamma_d$ determines how $i$ should act during both the communication phase and the action phase. During the communication phase, $i$'s strategy tells $i$ which messages it should send to other players or to the mediator given its local history up to the current point, and also which internal computations $i$ should perform. Afterwards, during the action phase, $\sigma_i$ tells $i$ which action to play given its local history.

A strategy profile $\vec{\sigma} = (\sigma_1, \ldots, \sigma_n)$ for $\Gamma_{CT}$ (resp., $\vec{\sigma} + \sigma_d$ for $\Gamma_d$, where $\sigma_d$ is the mediator's strategy) induces a map $\vec{\mu}_{\vec{\sigma}} : T \longrightarrow \Delta(A_i)$ (resp., $\vec{\mu}_{\vec{\sigma} + \sigma_d} : T \longrightarrow \Delta(A_i)$)  where each type profile $\vec{t}$ is mapped to a distribution $p$ over action profiles determined by the likelihood of each possible outcome of the game when players $\vec{\sigma}$ (resp., $\vec{\sigma} + \sigma_d$) with type profile $\vec{t}$. More precisely, $p(\vec{a})$ is defined as the probability that players end up  playing action profile $\vec{a}$ during the action phase when they play $\vec{\sigma}$ (resp., $\vec{\sigma} + \sigma_d$) with type profile $\vec{t}$. We call $\vec{\mu}_{\vec{\sigma}}$ the \emph{outcome} of $\vec{\sigma}$.  The expected utility $u_i(\vec{\sigma})$ (resp., $u_i(\vec{\sigma} + \sigma_d)$) of player $i$ when players play strategy profile $\vec{\sigma}$ (resp., $\vec{\sigma} + \sigma_d$) is given by $u_i(\vec{\mu}_{\vec{\sigma}})$, which is computed as in Equation~\ref{eq:expected-payoff}.

As in Definition~\ref{def:bayesian-nash}, a strategy profile $\vec{\sigma}$ for $\Gamma_{CT}$ (resp., $\vec{\sigma} + \sigma_d$ for $\Gamma_d)$ is a Nash equilibrium if no player can increase its expected utility by switching to a different strategy. Note that, in $\Gamma_d$, the mediator has no incentives and thus can never increase its utility by defecting from the proposed strategy.

The following example illustrates all these concepts and shows how adding a communication phase may affect the utilities of the players.

\begin{example}
Consider a game $\Gamma = (P, T, q, A, U)$ where $P = \{1,2\}$, $T_1 = T_2 = A = \{0,1\}$, $q$ is the uniform distribution over $\{0,1\}^2$ and $$u_i(\vec{t}, \vec{a}) = \left\{
\begin{array}{ll}
1 & \text{if } a_i = t_{3-i}\\
0 & \text{otherwise}
\end{array}
\right.$$
for all $i \in \{1,2\}$. 

In $\Gamma$, all strategy profiles $(\mu_1, \mu_2)$ are Nash equilibria of $\Gamma$ that give $1/2$ utility to each player in expectation (note that, regardless of what player 1 plays, the probability that it matches $2$'s type is $1/2$). However, if we extend $\Gamma$ to $\Gamma_{CT}$ or $\Gamma_d$, there are additional strategies that give different outcomes. For instance, consider a strategy profile $\vec{\sigma}$ in $\Gamma_{CT}$ in which players send each other a message with their type in the first round and then play whatever they received in the second round. It is easy to check that this strategy is a Nash equilibrium of $\Gamma_{CT}$ that gives $1$ utility to each player.
\end{example}

\subsection{The Revelation Principle}

As a consequence of the revelation principle~\cite{myerson79, myerson1981optimal}, every Bayesian Nash equilibrium of $\Gamma_d$ or $\Gamma_{CT}$ can be implemented in $\Gamma_{CT}$ using very simple communication protocols. More precisely, given a Bayesian game $\Gamma = (P, T, q, A, U)$ and a function $\mu: T \longrightarrow \Delta A$, consider a strategy profile $\vec{\tau}^\mu + \tau_d^\mu$ in which the players and the mediator act as follows:
\begin{enumerate}
\item During the first round of communication, each player $i$ sends its type $t_i$ to the mediator.
\item During the second round, the mediator receives the types sent by the players during the first round, it samples $\vec{a} \longleftarrow \mu(\vec{t})$ and sends $a_i$ to each player $i$. If a player $i$ didn't send its type, the mediator acts as if $i$ sent some default value in $T_i$.
\item During the third round, each player $i$ plays action $a_i$, where $a_i$ is the value received from the mediator.
\end{enumerate}

Intuitively, in $\vec{\tau}^\mu + \tau_d^\mu$, all players send their types to the mediator, the mediator samples an action profile from $\mu(\vec{t})$, and then the players play whatever the mediator suggested. We denote $\vec{\tau}^\mu + \tau_d^\mu$ as the $\mu$-canonical strategy profile of $\Gamma_d$. We also say that a strategy profile is canonical if it is $\mu$-canonical for some $\mu$. The following two propositions are the result of applying the revelation principle to mediated communication:

\begin{proposition}\label{prop:canonical}
If $\Gamma = (P, T, q, A, U)$ is a Bayesian game and $\vec{\sigma} + \sigma_d$ is a Nash equilibrium of $\Gamma_d$, then $\vec{\tau}^{\mu_{\vec{\sigma} + \sigma_d}} + \tau_d^{\mu_{\vec{\sigma} + \sigma_d}}$ is also a Nash equilibrium of $\Gamma_d$ that induces the same outcome.
\end{proposition}

\begin{proposition}~\label{prop:canonical2}
If $\Gamma = (P, T, q, A, U)$ is a Bayesian game and $\vec{\sigma}$ is a Nash equilibrium of $\Gamma_{CT}$, then $\vec{\tau}^{\mu_{\vec{\sigma}}} + \tau_d^{\mu_{\vec{\sigma}}}$ is a Nash equilibrium of $\Gamma_d$ that induces the same outcome.
\end{proposition}

Propositions~\ref{prop:canonical} and \ref{prop:canonical2} state that all equilibria in $\Gamma_d$ and $\Gamma_{CT}$ can be implemented using canonical strategies. In particular, this means that if we want to find specific equilibria in $\Gamma_d$ that satisfy certain properties, we can always restrict our search to canonical strategy profiles.

\subsection{Information Transmission Games}\label{sec:persuasion-games}

Information transmission games are a special type of Bayesian games for two players $s$ and $r$ in which the utilities of both players depend uniquely on the type of $s$ and the action played by $r$. Players $s$ and $r$ are typically called \emph{sender} and \emph{receiver}. Intuitively, information transmission games can be seen as scenarios in which the sender has complete knowledge of the game state and the receiver must make a decision with the sender's help. In both the mediator and the long cheap talk extensions of $\Gamma$, the sender must try to influence the behavior of the receiver by strategically disclosing information about her type. For instance, as in Kamenica and Gentzkow's example~\cite{kamenica2011bayesian}, the receiver could be a judge that has to decide if a certain person is guilty or not, and the sender could be the prosecutor of the accused. In this example, the prosecutor's aim is to convince the judge that the accused is guilty.



\subsection{Welfare and the Value of Mediation}

We extend Ashlagi, Monderer and Tennenholtz's definition of the \emph{value of mediation}~\cite{ashlagi2008value} as follows. Let $w : T \times A \rightarrow \mathbb{R}^+$ be a function that takes as input a type profile $\vec{t}$ and an action profile $\vec{a}$ and outputs the total welfare of the system given that players had type $\vec{t}$ and played $\vec{a}$. Given a strategy profile $\vec{\sigma}$ for a Bayesian game $\Gamma$, denote by $w(\vec{\sigma})$ the expected welfare when players play $\vec{\sigma}$, which is computed in the same way as in Equation~\ref{eq:expected-payoff}. Moreover, let $NE^{w}(\Gamma)$ denote the maximum expected value of $w$ on a Nash Equilibrium of $\Gamma$ and let $NE^{w}(\Gamma_{CT})$ and $NE^{w}(\Gamma_d)$ be defined analogously. Given the welfare function $w$, we define the \emph{value of mediation} of $\Gamma$ (denoted by $M^w(\Gamma)$) as the ratio between the maximum welfare that players can achieve in equilibrium with the mediator and the maximum welfare they can reach in equilibrium by communicating without the mediator. More precisely, $$M^w(\Gamma) := \frac{NE^{w}(\Gamma_{d})}{NE^{w}(\Gamma_{CT})}.$$

In short, $M^w$ compares the welfare that the system can get with a mediator as opposed to the welfare that the system can get by allowing players to communicate with long cheap talk. It is important to note that the definition above is not well defined when $NE^{w}(\Gamma_{CT}) = 0$. Thus, we extend the definition as follows: if $NE^{w}(\Gamma_{CT}) = 0$, then $$M^w(\Gamma) = \left\{
\begin{array} {ll}
1 & \text{if } NE^{w}(\Gamma_{d}) = 0\\
+\infty & \text{otherwise.}
\end{array}\right.
$$

It is important to note that, by Proposition~\ref{prop:canonical2}, $NE^{w}(\Gamma_{d}) \ge NE^{w}(\Gamma_{CT}) \ge 0$ for all positive welfare functions $w$. This means that $$M^w(\Gamma) \ge 1$$ for all positive welfare functions $w$.

\section{Main Results}

In this section we state our main results regarding the value of mediation for different types of information transmission games. Our first result states that the mediator cannot improve the utility of the sender or the utility of the receiver if the game has binary actions. More generally, the following theorem states that the mediator cannot improve any welfare function $w$ that is \emph{monotone}  over $u_r$ and $u_s$, which is a function such that $$(u_r(\omega, a) \ge u_r(\omega', a') \quad \text{and} \quad u_s(\omega, a) \ge u_s(\omega', a')) \quad \Longrightarrow \quad w(\omega, a) \ge w(\omega', a').$$

\begin{theorem}\label{thm:persuasion-binary-core}
If $\Gamma = (P, T, q, A, U)$ is an information transmission game such that $|A_r| \le 2$ and $w$ is any positive welfare function for $\Gamma$ that is monotone over $u_r$ and $u_s$, then $$M^{w}(\Gamma) = 1.$$
\end{theorem}

Theorem~\ref{thm:persuasion-binary-core} requires that $|A_r| \le 2$ and applies to all positive welfare functions that are monotone over $u_r$ and $u_s$. In particular, it applies to $u_r$ and $u_s$ themselves  and all positive linear combinations of these. We show next that these constraints are tight, which means that there exist information transmission games with three (or more) actions and welfare functions that are not monotone over $u_r$ and $u_s$ such that Theorem~\ref{thm:persuasion-binary-core} does not hold. Regarding the monotonicity, consider the following example:

\begin{example}\label{example:non-monotone}
Let $\Gamma = (P, T, q, A, U)$ be an information transmission game such what $T_s = \{t_0, t_1\}$, $A_r = \{0,1\}$, $q$ is the uniform distribution over $\{t_0,t_1\}$, and the utilities of the sender and the receiver are defined as follows:

$$\begin{array}{c|cc}
u_s & 0 & 1\\
\hline
t_0 & 1 & 0\\
t_1 & 0 & 1\\
\end{array}
\quad  \quad \quad 
\begin{array}{c|cc}
u_r & 0 & 1\\
\hline
t_0 & 1 & 0\\
t_1 & \frac{1}{2} & 1\\
\end{array}
$$

Consider also a welfare function $w$ defined by the following table

$$\begin{array}{c|cc}
w & 0 & 1\\
\hline
t_0 & 0 & 1\\
t_1 & 0 & 0\\
\end{array}
$$
\end{example}

Intuitively, $\Gamma$ is an information transmission game with two states and two actions in which players have similar preferences. However, the welfare function $w$ is positive in the only combination of state and action in which none of the players get any utility. It is possible to show that the receiver would never play action $1$ on state $t_0$ if players could communicate with long cheap-talk. However, such a combination is possible with a mediator. This implies the following result, which is proven in Section~\ref{sec:proof-2}.

\begin{proposition}\label{prop:example1}
Let $\Gamma$ and $w$ be the information transmission game and welfare function of Example~\ref{example:non-monotone}. Then, $$M^{w}(\Gamma) = +\infty.$$
\end{proposition}

Note that, given a Bayesian game $\Gamma$ and a welfare function $w$ such that $M^{w}(\Gamma) = +\infty$, we can define another welfare function $w' := w + \epsilon$ that gives $\epsilon$ more utility on all inputs. By setting $\epsilon$ to arbitrarily small values, we can construct new welfare functions that produce arbitrarily high (non-infinite) mediation values. This, in addition to Proposition~\ref{prop:example1}, gives the following corollary.

\begin{corollary}
For all $N > 0$, there exist information transmission games $\Gamma = (P, T, q, A, U)$ with $|A_r| = 2$ and positive welfare functions $w$ for $\Gamma$ such that 
$$M^{w}(\Gamma)  > N.$$
\end{corollary}

Regarding the bound on the number of actions of Theorem~\ref{thm:persuasion-binary-core}, consider the following example.

\begin{example}\label{example:three-actions}
Let $\Gamma = (P, T, q, A, U)$ be an information transmission game such what $T_s = \{t_0, t_1, t_2\}$, $A_r = \{0,1,2\}$, $q$ is the uniform distribution over $\{t_0, t_1, t_2\}$, and the utilities of the sender and the receiver are defined as follows:

$$\begin{array}{c|ccc}
u_s & 0 & 1 & 2\\
\hline
t_0 & 1 & 0 & 0\\
t_1 & 0 & 1 & 0\\
t_2 & 0 & 0 & 1\\
\end{array}
\quad  \quad \quad 
\begin{array}{c|ccc}
u_r & 0 & 1 & 2\\
\hline
t_0 & 0 & 1 & 0\\
t_1 & 0 & 0 & 1\\
t_2 & 1 & 0 & 0\\
\end{array}
$$

Consider also a $\mu$-canonical strategy profile $\vec{\tau}^\mu + \tau_d^\mu$ for $\Gamma_d$, where $\mu$ is the map that maps each state $t_i$ to the uniform distribution over $\{t_i, t_{i+1}\}$ (using the convention that $t_{2+1} = t_0$).
\end{example}

In this example there are three actions and three states, and the action preferred in each state is different for each player. The outcome $\mu$ of the proposed strategy selects the preferred action of the sender and that of the receiver with equal probability. We can show that $\vec{\tau}^\mu + \tau_d^\mu$ is a Nash equilibrium of $\Gamma_d$ that strictly Pareto-dominates all Nash equilibria in $\Gamma_{CT}$:

\begin{proposition}\label{prop:example2}
Let $\Gamma$ and $\mu$ be the information transmission game and outcome of Example~\ref{example:three-actions}. Then, $\vec{\tau}^\mu + \tau_d^\mu$ is a Nash equilibrium of $\Gamma_d$ such that, for all strategy profiles $\vec{\sigma}$ in $\Gamma_{CT}$, $$u_s(\vec{\sigma}) < u_s(\vec{\tau}^\mu + \tau_d^\mu) \quad \quad \mbox{and} \quad \quad u_r(\vec{\sigma}) < u_r(\vec{\tau}^\mu + \tau_d^\mu).$$
\end{proposition}

\section{Proof of Theorem~\ref{thm:persuasion-binary-core}}

At a high level, proof of Theorem~\ref{thm:persuasion-binary-core} goes as follows. First we show that, in information transmission games with binary actions there exist non-constant Nash equilibria in $\Gamma_d$ if and only if the sender and the receiver are \emph{aligned}, which happens when the receiver is better off in expectation playing exactly what the sender wants instead of playing an action with no information. In this case, it is straightforward to check that the equilibrium that maximizes both the sender and the receiver utilities is precisely the one in which the receiver plays the optimal action for the sender in every possible state. Moreover, it is possible to reach this outcome in both $\Gamma_d$ and $\Gamma_{CT}$, which implies that $NE^w(\Gamma_{CT}) = NE^w(\Gamma_d)$ for all functions $w$ that are monotone over $u_r$ and $u_s$.

Next, we provide the full proof. Given a function $\mu: T_s \longrightarrow \Delta(\{0,1\})$, let $P^\mu : T_s \longrightarrow [0,1]$ be the function that maps $\omega \in T_s$ to the probability that the distribution $\mu(\omega)$ assigns to $0$. We begin the proof of Theorem~\ref{thm:persuasion-binary-core} by characterizing all outcomes of Nash equilibria of $\Gamma_d$.

\begin{proposition}\label{prop:characterization-binary}
Let $\Gamma = (P,T,q,A,U)$ be an information transmission game with $A_r = \{0,1\}$. Then, given a function $\mu: T_s \rightarrow \Delta(A_r)$, there exists a Nash equilibrium of $\Gamma_d$ that induces $\mu$ if and only if, for each type $\omega \in T_s$ we have that
\begin{itemize}
\item[(a)] $u_s(\omega, 0) > u_s(\omega,1) \Longrightarrow P^\mu(\omega) \ge P^\mu(\omega')$ for all $\omega' \in T_s$.
\item[(b)] $u_s(\omega, 0) < u_s(\omega,1) \Longrightarrow P^\mu(\omega) \le P^\mu(\omega')$ for all $\omega' \in T_s$.
\item[(c)] $u_r(\mu) \ge \max(u_r(\mathbf{0}_{T_s}),u_r(\mathbf{1}_{T_s}))$, where $\mathbf{0}_{T_s}$ and $\mathbf{1}_{T_s}$ are the functions that map all elements to the constant $0$ and $1$ distributions respectively.
\end{itemize}
\end{proposition}

This proposition says that a function $\mu: T_s \rightarrow \Delta(A_r)$ can be the outcome of some Nash equilibrium of $\Gamma_d$ if and only if (a) the types in which the sender prefers $0$ are also the ones in which it is more likely that the receiver plays $0$, (b) the types in which the sender prefers $1$ are also the ones in which it is more likely that the receiver plays $1$, and (c) the receiver gets a better utility from playing the suggested action than by ignoring the suggestion and playing the action that a priori gives her the most utility with no information.

\begin{proof}[Proof of Proposition~\ref{prop:characterization-binary}]
Let $\vec{\tau}^\mu + \sigma_d^\mu$ be the $\mu$-canonical strategy of $\Gamma_d$. We show independently that $\vec{\tau}^\mu + \sigma_d^\mu$ is incentive-compatible for the sender if and only if (a) and (b) are satisfied, and that $\vec{\tau}^\mu + \sigma_d^\mu$ is incentive-compatible for the receiver if (c) is satisfied.

Note that $\vec{\tau}^\mu + \sigma_d^\mu$ is incentive-compatible for the sender if and only if the sender is always incentivized to report its true type to the mediator. This implies that, the probability of playing $0$ should be maximal on types in which the sender prefers $0$ to $1$, and the probability of playing $1$ should be maximal (i.e., the probability of playing $0$ should be minimal) whenever the sender prefers $1$ to $0$. It is also straightforward to check that, if these conditions are satisfied, then $\vec{\tau}^\mu + \sigma_d^\mu$ is incentive-compatible for the sender.

For the receiver, note that if  $\vec{\tau}^\mu + \sigma_d^\mu$ is incentive-compatible, then condition (c) is necessary. Otherwise, the receiver is better off ignoring all the communication phase altogether and playing the action that reports the most utility in expectation without any knowledge of the sender's type. For the converse, suppose that condition (c) holds but that  $\vec{\tau}^\mu + \sigma_d^\mu$ is not incentive-compatible for the receiver. This means that the receiver either prefers to play $0$ when the mediator suggests $1$ or that the receiver prefers to play $1$ when the mediator suggests $0$ (or both). The first case implies that the receiver can increase her utility by playing $0$ regardless of the mediator's suggestion, while the second case implies that the receiver can increase her utility by always playing $1$. Both cases contradict the fact that $u_r(\mu) \ge \max(u_r(\mathbf{0}_{T_s}),u_r(\mathbf{1}_{T_s}))$.
\end{proof}

Suppose that in $\Gamma$ there are no types $\omega \in T_s$ such that $u_s(\omega, 0) = u_s(\omega, 1)$. Then, we can partition $T_s$ into two subsets $T_s^0$ and $T_s^1$, which contain the types in which the sender prefers action $0$ and action $1$ respectively. Conditions (a) and (b) of Proposition~\ref{prop:characterization-binary} imply that, if $\mu$ is the outcome of a Nash equilibrium of $\Gamma_d$, then there exist two values $p_0, p_1 \in [0,1]$ with $p_0 \ge p_1$ such that 
$$
P^\mu(\omega) = \left\{
\begin{array}{ll}
p_0 & \mbox{if } \omega \in T_s^0\\
p_1 & \mbox{otherwise.}
\end{array}
\right.
$$

With this notation, condition (c) is equivalent to the following system of equations: $$
\begin{array}{ll}
\sum_{\omega \in T_s^0} q(\omega)(p_0 u_r(\omega, 0) + (1 - p_0)u_r(\omega, 1)) + \sum_{\omega \in T_s^1}q(\omega)(p_1 u_r(\omega, 0) + (1 - p_1)u_r(\omega, 1)) \ge u_r(\mathbf{0}_{T_s})\\
\sum_{\omega \in T_s^0} q(\omega)(p_0 u_r(\omega, 0) + (1 - p_0)u_r(\omega, 1)) + \sum_{\omega \in T_s^1}q(\omega)(p_1 u_r(\omega, 0) + (1 - p_1)u_r(\omega, 1)) \ge u_r(\mathbf{1}_{T_s}),
\end{array}
$$

where $q(\omega)$ is the probability that the sender's type is $\omega$ according to the type distribution $q$. The right-hand side of the equations can be expressed as follows: $$
\begin{array}{l}
u_r(\mathbf{0}_{T_s}) = \sum_{\omega \in T_s^0} q(\omega)u_r(\omega, 0) + \sum_{\omega \in T_s^1} q(\omega)u_r(\omega, 0)\\
u_r(\mathbf{1}_{T_s}) = \sum_{\omega \in T_s^0} q(\omega)u_r(\omega, 1) + \sum_{\omega \in T_s^1} q(\omega)u_r(\omega, 1).
\end{array}$$

Given $i,j \in \{0,1\}$, let $A_i^j := \sum_{\omega \in T_s^j} q(\omega) u_r(\omega, i)$. Then, condition (c) can be expressed as $$
\begin{array}{ll}
p_0A_0^0 + (1-p_0)A_1^0 + p_1A_0^1 + (1 - p_1)A_1^1 \ge A_0^0 + A_0^1\\
p_0A_0^0 + (1-p_0)A_1^0 + p_1A_0^1 + (1 - p_1)A_1^1 \ge A_1^0 + A_1^1,
\end{array}
$$
which is equivalent to $$
\begin{array}{rl}
(1-p_0)(A_1^0- A_0^0) + (1 - p_1)(A_1^1 - A_0^1) & \ge 0\\
p_0(A_0^0 - A_1^0) + p_1(A_0^1 - A_1^1) & \ge 0.
\end{array}
$$
Setting $B_0 := A_0^0 - A_1^0$ and $B_1 = A_0^1 - A_1^1$ we get
$$
\begin{array}{rl}
(1 - p_0)B_0 + (1-p_1)B_1 & \le 0\\
p_0B_0 + p_1B_1 & \ge 0,\\
\end{array}
$$
which can be reduced to 
\begin{equation}\label{eq:solutions-binary}
p_0B_0 + p_1B_1 \ge \max(0, B_0 + B_1).
\end{equation}

Thus, $\mu$ is the outcome of a Nash equilibrium of $\Gamma_d$ if and only if $0 \le p_1 \le p_0 \le 1$ and the inequality above is satisfied. Note that $B_0$ and $B_1$ are constants that depend uniquely on the type distribution $q$ and the utility functions $u_r$ and $u_s$. The following proposition characterizes all solutions of outcomes of Nash equilibria of $\Gamma_d$ when $B_0 < 0$ or $B_1 > 0$.

\begin{proposition}\label{prop:constant-solutions}
If $\Gamma = (P,T,q,A,U)$ is an information transmission game such that $B_0 < 0$ or $B_1 > 0$, then, all outcomes of Nash equilibria of $\Gamma_d$ are constant (i.e., they map all types to the same distribution).
\end{proposition}

\begin{proof}
Proposition~\ref{prop:constant-solutions} reduces to show that, if $B_0 < 0$ or $B_1 > 0$, all solutions of Equation~\ref{eq:solutions-binary} such that $0 \le p_1 \le p_0 \le 1$ satisfy $p_0 = p_1$. We divide the proof in the following cases: 
\begin{itemize}
\item \textbf{Case $B_1 > 0$ and $B_0 \ge 0:\quad$} $p_0B_0 + p_1B_1 \ge B_0 + B_1$ is satisfied only if $p_1 = 1$. Therefore, the only solution such that $0 \le p_1 \le p_0 \le 1$ is $p_0 = p_1 = 1$.
\item \textbf{Case $B_1 > 0$ and $B_0 < 0$:} 
\begin{itemize}
\item \textbf{Case $B_0 + B_1 < 0$:} If $B_0 + B_1 < 0$, then $p_0B_0 + p_1B_1 \ge 0$ must be satisfied. However, $p_0B_0 + p_1B_1 \le p_0(B_0 + B_1) \le 0$, and the equality only holds if $p_0 = 0$. Therefore, in this case, the only solution is $p_0 = p_1 = 0$. 
\item \textbf{Case $B_0 + B_1 > 0$:} If $B_0 + B_1 > 0$, then it is required that $p_0B_0 + p_1B_1 \ge B_0 + B_1$, which is equivalent to $(1 - p_0)B_0 + (1-p_1)B_1 \le 0$. However, $(1 - p_0)B_0 + (1-p_1)B_1 \ge (1-p_1)(B_0 + B_1) \ge 0$, and the last equality holds only when $p_1 = 1$. Thus, the only solution is $p_0 = p_1 = 1$. 
\item \textbf{Case $B_0 + B_1 = 0$:} If $B_0 = -B_1$, then we need that $(p_0 - p_1)B_0 \ge 0$, which, since $p_0 \ge p_1$, can only happen if $p_0 = p_1$.
\end{itemize}
\item \textbf{Case $B_1 \le 0$ and $B_0 < 0:\quad$} $p_0B_0 + p_1B_1 \ge 0$ is satisfied only if  $p_0 = 0$, which means that the only solution is $p_0 = p_1 = 0$.
\end{itemize}

In short, if $B_1 > 0$ and $B_0 + B_1 > 0$, the only solution is $p_0 = 1, p_1 = 1$. If $B_0 < 0$ and $B_0 + B_1 < 0$, the only solution is $p_0 = 0, p_1 = 0$. The last remaining case is when $B_0 < 0$,  $B_1 > 0$, and $B_0 + B_1 = 0$. Here, the only solutions are those with $p_0 = p_1$. A visual representation of the solutions is shown in Figure~\ref{fig:solutions-binary}.

\begin{center}
\begin{figure}[H]
\label{fig:solutions-binary}
\centering
\includegraphics[scale=1]{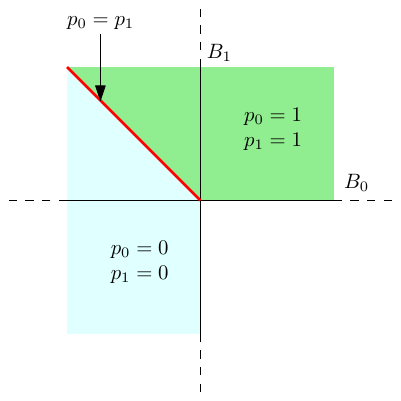}
\caption{Visual representation of all solutions of Equation~\ref{eq:solutions-binary} with $0 \le p_1 \le p_0 \le 1$.}
\end{figure}
\end{center}
\end{proof}

Proposition~\ref{prop:constant-solutions} shows that, if $B_0 < 0$ or $B_1 > 0$, all outcomes of Nash equilibria of $\Gamma_d$ are constant. Suppose that $\mu : T_s \longrightarrow \Delta(A_r)$ is a constant function that is the outcome of a Nash equilibrium of $\Gamma_d$. Then, it is easy to check that, in $\Gamma_{CT}$, the strategy profile in which the sender does nothing and the receiver simply plays an action sampled from the image of $\mu$ is a Nash equilibrium as well. Therefore, Proposition~\ref{prop:constant-solutions} implies the following corollary.

\begin{corollary}\label{cor:constant}
If $\Gamma = (P,T,q,A,U)$ is an information transmission game such that $B_0 < 0$ or $B_1 > 0$ and $w$ is a positive welfare function , then $$M^w(\Gamma) = 1.$$
\end{corollary}

This result may seem stronger than Theorem~\ref{thm:persuasion-binary-core} since it does not require that $w$ is monotone over $u_r$ and $u_s$. However, it only applies if $B_0 < 0$ or $B_1 > 0$. In fact, Proposition~\ref{prop:example1} shows that Corollary~\ref{cor:constant}  does not hold in general. The remaining case is the one in which $B_0 \ge 0$ and $B_1 \le 0$. It is easy to check that $B_0$ is positive if and only if the receiver gets an equal or better expected utility playing 0 than playing 1 conditioned on the fact that the state is in $T_s^0$. Analogously, $B_1$ is negative if the receiver prefers playing $1$ than $0$ in $T_s^1$. Therefore, if $B_0 \ge 0$ and $B_1 \le 0$ we say that the sender and the receiver are \emph{aligned} since they have approximately the same preferences. The following proposition states that, in this case, the outcome of a Nash equilibrium that maximizes both the sender and the receivers' utilities is the one in which the receiver plays the best action for the sender in every single state.

\begin{proposition}\label{prop:best-eq}
If $\Gamma = (P,T,q,A,U)$ is an information transmission game such that $B_0 \ge 0$ and $B_1 \le 0$, then the outcome defined by $p_0 = 1$ and $p_1 = 0$ is incentive-compatible for the sender and for the receiver, and maximizes both $u_r$ and $u_s$ among all Nash equilibria of $\Gamma_d$.
\end{proposition}

\begin{proof}
It is straightforward to check that $p_0 = 1$ and $p_1 = 0$ satisfies Equation~\ref{eq:solutions-binary} if the sender and the receiver are aligned. Moreover, in this outcome, the sender gets the maximum possible utility in each state. It remains to prove that this outcome is also the best for the receiver. To see this, note that $$u_r(\mu) = p_0A_0^0 + (1 - p_0)(A_1^0) + p_1 A_0^1 + (1 - p_1)A_1^1 = A_1^0 + A_1^1 + p_0B_0 + p_1B_1.$$
Since $B_0 \ge 0$ and $B_1 \le 0$, this expression is maximized when $p_0 = 1$ and $p_1 = 0$.
\end{proof}

Proposition~\ref{prop:best-eq} shows that, if the sender and the receiver are \emph{aligned}, then the outcome of a Nash equilibrium that maximizes both $u_s$ and $u_r$ is the one in which the receiver plays the sender's preferred action in each state. In particular, this outcome $\mu$ maximizes any positive welfare function $w$ that is monotone over $u_s$ and $u_r$. To prove Theorem~\ref{thm:persuasion-binary-core}, it remains to show that there exists a Nash equilibrium in $\Gamma_{CT}$ that gives the same welfare in expectation. For this purpose, consider the strategy profile $\vec{\sigma}$ in $\Gamma_{CT}$ in which the sender sends the receiver a message with its preferred action and the receiver plays that action. By construction, $\vec{\sigma}$ induces $\mu$ and, since the sender and the receiver are aligned, $\vec{\sigma}$ is a Nash equilibrium of $\Gamma_{CT}$. This proves Theorem~\ref{thm:persuasion-binary-core}.

It is important to note that, for ease of exposition, we were assuming that there were no states in which the receiver gets the same utility with action $0$ and with action $1$. However, an analogous reasoning applies for these cases as well, as shown in the following proposition.

\begin{proposition}\label{prop:characterization-binary-extended}
If $\Gamma = (P,T,q,A,U)$ is an information transmission game and $w$ is a positive welfare function that is monotone over $u_r$ and $u_s$, then there exists a Nash equilibrium $\vec{\sigma} + \sigma$ of $\Gamma_d$ that maximizes $w$ such that there exist two values $p_0, p_1 \in [0,1]$ with $p_0 \ge p_1$ such that 
$$
P^{\mu_{\vec{\sigma} + \sigma}}(\omega) = \left\{
\begin{array}{ll}
p_0 & \mbox{if } u_s(0) > u_s(1) \mbox{ or } (u_s(0) = u_s(1) \mbox{ and } u_r(0) > u_r(1))\\
p_1 & \mbox{otherwise.}
\end{array}
\right.
$$
\end{proposition}

Proposition~\ref{prop:characterization-binary-extended} implies that the proof of Theorem~\ref{thm:persuasion-binary-core} in the general case is analogous to the proof in the case where the sender is never indifferent between $0$ and $1$. The only difference is that, in this case, we define $T_s^0$ as the subset of types $\omega$ in which either the sender strictly prefers $0$ to $1$ or the sender is indifferent and the receiver strictly prefers $0$, and we define $T_s^1$ as the subset of the remaining types.

\begin{proof}[Proof of Proposition~\ref{prop:characterization-binary-extended}]

Suppose that $\mu$ is the outcome of a Nash equilibrium of $\Gamma_d$ that maximizes $w$. By Proposition~\ref{prop:characterization-binary}, there exist two values $p_0$ and $p_1$ with $p_0 \ge p_1$ such that $P^\mu(\omega) = p_0$ if the sender prefers $0$ to $1$ in $\omega$, $P^\mu(\omega) = p_1$ if the sender prefers $1$ to $0$, and $P^\mu(\omega) \in [p_1, p_0]$ if the sender is indifferent.
Consider a function $P^{\mu'}$ that is equal to $P^\mu$ at all types $\omega$ in which the sender strictly prefers $0$ or in which the sender strictly prefers $1$. In the types $\omega$ in which the sender is indifferent between $0$ and $1$, we set $P^{\mu'}(\omega) = p_0$ if the receiver prefers $0$, and $P^{\mu'}(\omega) = p_1$ if the receiver prefers $1$ or is indifferent between $0$ and $1$. It is easy to check that $\mu'$ is the outcome of a Nash equilibrium of $\Gamma_d$ since, if $P^\mu$ satisfies the conditions of Proposition~\ref{prop:characterization-binary}, so does $P^{\mu'}$. Moreover $u_r(\mu') \ge u_r(\mu)$ while $u_s(\mu') = u_s(\mu)$. Since $w$ is positive and monotone over $u_r$ and $u_s$, $P^{\mu'}$ also maximizes $w$ as desired.
\end{proof}

\section{Proof of Proposition~\ref{prop:example1}}\label{sec:proof-2}

Consider a $\mu$-canonical strategy profile $\vec{\tau}^\mu + \tau_d^\mu$ in $\Gamma_d$, where $\mu$ is the function that maps $t_1$ to to the constant distribution that assigns $1$ with probability $1$, and maps $t_0$ to the uniform distribution over $\{0,1\}$. It is easy to check that $\mu$ satisfies the conditions of Proposition~\ref{prop:characterization-binary}, and therefore $\vec{\tau}^\mu + \tau_d^\mu$ is a Nash equilibrium of $\Gamma_d$ such that $w(\vec{\tau}^\mu + \tau_d^\mu) = \frac{1}{4}$. Proposition~\ref{prop:example1} follows from the following proposition.

\begin{proposition}
There is no equilibrium in $\Gamma_{CT}$ in which, with positive probability, the receiver plays $1$ when the sender has type $t_0$.
\end{proposition}

\begin{proof}
Suppose that $\vec{\sigma}$ is a Nash equilibrium of $\Gamma_{CT}$ in which, with positive probability, the receiver plays $1$ when the sender has type $t_0$. Let $H$ be a global history (i.e., a pair $(h_s, h_r)$ of the local histories of the sender and the receiver, respectively) that could occur with positive probability when playing $\vec{\sigma}$ and in which the sender's type is $t_0$ and the receiver plays action $1$. Since $\vec{\sigma}$ is a Nash equilibrium, the expected utility of the receiver when playing $1$ conditioned to the fact that its local history is $h_r$ must be greater than its expected utility when playing $0$. This means that 
$$p_1(h_r) \ge p_0(h_r) + \frac{1}{2} p_1(h_r),$$

where $p_i(h_r)$ denotes the probability that the sender's type is $t_i$ conditioned that the receiver's local history is $h_r$ and both the sender and the receiver are following $\vec{\sigma}$. Using the fact that $p_0(h_r) = 1 - p_1(h_r)$ we get that $$p_1(h_r) \ge \frac{2}{3}.$$

Suppose that $h_r$ has a total of $n$ rounds of communication. Let $h_r^0, h_r^1, \ldots, h_r^n$ be the prefixes of $h_r$ by the end of rounds $0, 1, \ldots, n$ respectively. Since $h_r^0$ is an empty history, we have that $p_1(h_r^0) = \Pr[t_1 \longleftarrow q] = \frac{1}{2}$. Moreover, by assumption, $p_1(h_r^n) \ge \frac{2}{3}$. Therefore, there exists an integer $k < n$ such that $p_1(h_r^k) < p_1(h_r^{k+1})$. This means that there exists a round $k$ in which the information sent between the sender and the receiver at round $k$ increases the likelihood that the sender's type is $t_1$. However, note that the relationship between $p_1(h_r^{k+1})$ and $p_1(h_r^k)$ does not depend on the messages sent by the receiver on round $k+1$ since these are just generated as a function of $h_r^k$. Therefore, the increase in probability between round $k$ and $k+1$ is due exclusively to the information sent by the sender, but this contradicts the fact that $\vec{\sigma}$ is a Nash equilibrium. More precisely, since there are only two players, the sender and the receiver know each other's local history, and thus, since the sender prefers action $0$ over action $1$ when its type is $t_0$, the sender will always choose to send the messages in round $k$ that minimize $p_1(h_r^{k+1})$. By definition of $p_1$, the sender can always choose a combination of messages in round $k$ such that $p_1(h_r^{k+1})$ does not increase with respect to $p_1(h_r^k)$. Therefore, in a Nash equilibrium, the sender would never send a combination of messages such that $p_1(h_r^k) < p_1(h_r^{k+1})$.
\end{proof}

\section{Proof of Proposition~\ref{prop:example2}}

We begin the proof by showing that $\vec{\tau}^\mu + \tau_d^\mu$ is a Nash equilibrium of $\Gamma_d$ such that both the sender and the receiver get $1/2$ expected utility each.

\begin{proposition}\label{prop:neq-mediator-nonbinary}
$\vec{\tau}^\mu + \tau_d^\mu$ is a Nash equilibrium of $\Gamma_d$ such that $$u_s(\vec{\tau}^\mu + \tau_d^\mu) = u_r(\vec{\tau}^\mu + \tau_d^\mu) = \frac{1}{2}.$$
\end{proposition}

\begin{proof}
First we check that $\vec{\tau}^\mu + \tau_d^\mu$ is incentive-compatible for the sender. Indeed, if the sender reports its true type, it will always get an expected utility of $1/2$, which is at least as good as the expected utility when reporting any other state. To check that $\vec{\tau}^\mu + \tau_d^\mu$ is incentive-compatible for the receiver, note that if the sender sent its true type and the receiver suggests action $i$, then there is $1/2$ chance that the sender's type is $t_i$ and $1/2$ chance that the sender's type is $t_{i-1}$ (using the convention that $t_{0-1} = t_2$). Therefore, playing action $i$ is optimal.
\end{proof}

Proposition~\ref{prop:example2} follows from Proposition~\ref{prop:neq-mediator-nonbinary} and the following proposition:

\begin{proposition}\label{prop:neq-ct-nonbinary}
Let $\vec{\sigma}$ be any Nash equilibrium of $\Gamma_{CT}$, then $u_r(\vec{\sigma}) < \frac{1}{2}$ and $u_s(\vec{\sigma}) < \frac{1}{2}$.
\end{proposition}

To prove Proposition~\ref{prop:neq-ct-nonbinary}, we begin by showing the following two lemmas. The first one states that the sender and the receiver always get the same expected utility in equilibrium.

\begin{lemma}\label{lemma-eq-1}
Let $\vec{\sigma}$ be any Nash equilibrium of $\Gamma_{CT}$, then $u_r(\vec{\sigma}) = u_s(\vec{\sigma})$.
\end{lemma}

The second lemma states that if, in some equilibrium $\vec{\sigma}$, for some $i \in \{0,1,2\}$, there exists a chance that the receiver plays action $i-1$ when the sender's type is $t_i$ (again, using the convention that $0-1 = 2$), then $u_r(\vec{\sigma}) < \frac{1}{2}$ and $u_s(\vec{\sigma}) < \frac{1}{2}$.

\begin{lemma}\label{lemma-eq-2}
Given a Nash equilibrium $\vec{\sigma}$ of $\Gamma_{CT}$, let $p_{i,j}(\vec{\sigma})$ denote the probability that the sender's type is $t_i$ and the receiver plays action $j$ when playing $\vec{\sigma}$. Then, $$(p_{0,2}(\vec{\sigma}) > 0 \mbox{ or } p_{1,0}(\vec{\sigma}) > 0 \mbox{ or } p_{2,1}(\vec{\sigma}) > 0) \Longrightarrow \left(u_r(\vec{\sigma}) < \frac{1}{2} \mbox{ and } u_s(\vec{\sigma}) < \frac{1}{2}\right)$$
\end{lemma}

We first prove Lemma~\ref{lemma-eq-1} and use it to prove Lemma~\ref{lemma-eq-2}.

\begin{proof}[Proof of Lemma~\ref{lemma-eq-1}]
Suppose that there exists a Nash equilibrium $\vec{\sigma}$ in $\Gamma_{CT}$ such that $u_s(\vec{\sigma}) < u_r(\vec{\sigma})$. Consider the strategy $\tau_s$ for the sender that consists of, given type $t_i$ to play exactly like $\sigma_s$ with type $t_{i-1}$. By construction, using the notation of Lemma~\ref{lemma-eq-2}, we have the following:
$$\begin{array}{c}
p_{0,0}((\tau_s,\sigma_r)) = p_{2, 0}(\vec{\sigma})\\
p_{1,1}((\tau_s,\sigma_r)) = p_{0, 1}(\vec{\sigma})\\
p_{2,2}((\tau_s,\sigma_r)) = p_{1, 2}(\vec{\sigma}),\\
\end{array}
$$
which means that 
$$p_{0,0}((\tau_s,\sigma_r)) + p_{1,1}((\tau_s,\sigma_r)) + p_{2,2}((\tau_s,\sigma_r)) =  p_{2, 0}(\vec{\sigma}) + p_{0, 1}(\vec{\sigma}) + p_{1, 2}(\vec{\sigma}), $$
and thus $u_s((\tau_s, \sigma_r)) = u_r(\vec{\sigma}) > u_s(\vec{\sigma})$. This contradicts the assumption that $\vec{\sigma}$ is a Nash equilibrium.

Suppose instead that there exists a Nash equilibrium $\vec{\sigma}$ in $\Gamma_{CT}$ such that $u_s(\vec{\sigma}) > u_r(\vec{\sigma})$. Consider the strategy $\tau_r$ for the receiver that consists of playing exactly like $\sigma_r$ except that, if it would play action $i$ in $\sigma_r$, it plays $i+1$ instead. By construction, we have that 
$$\begin{array}{c}
p_{0,0}(\vec{\sigma}) = p_{0, 1}(\sigma_s, \tau_r)\\
p_{1,1}(\vec{\sigma}) = p_{1, 2}(\sigma_s, \tau_r)\\
p_{2,2}(\vec{\sigma}) = p_{2, 0}(\sigma_s, \tau_r),\\
\end{array}
$$
which means that $u_r(\sigma_s, \tau_r) = u_s(\vec{\sigma}) > u_r(\vec{\sigma})$. Again, this contradicts the assumption that $\vec{\sigma}$ is a Nash equilibrium.
\end{proof}

\begin{proof}[Proof of Lemma~\ref{lemma-eq-2}]
Suppose that $p_{0,2}(\vec{\sigma}) > 0$, $p_{1,0}(\vec{\sigma}) > 0$, or $p_{2,1}(\vec{\sigma}) > 0$. Since $u_s(\vec{\sigma}) = p_{0,0}(\vec{\sigma}) + p_{1,1}(\vec{\sigma}) + p_{2,2}(\vec{\sigma})$ and $u_r(\vec{\sigma}) = p_{0,1}(\vec{\sigma}) + p_{1,2}(\vec{\sigma}) + p_{2,0}(\vec{\sigma})$, we have that $$u_s(\vec{\sigma}) + u_r(\vec{\sigma}) = 1 - (p_{0,2}(\vec{\sigma}) + p_{1,0}(\vec{\sigma}) + p_{2,1}(\vec{\sigma})).$$
By assumption, $p_{0,2}(\vec{\sigma}) + p_{1,0}(\vec{\sigma}) + p_{2,1}(\vec{\sigma}) > 0$, and therefore $$u_s(\vec{\sigma}) + u_r(\vec{\sigma}) < 1.$$
By Lemma~\ref{lemma-eq-1}, this implies that $u_r(\vec{\sigma}) < \frac{1}{2}$ and $u_s(\vec{\sigma}) < \frac{1}{2}$.
\end{proof}

With these lemmas, we are ready to tackle Proposition~\ref{prop:neq-ct-nonbinary}. We start by showing the following.

\begin{proposition}\label{prop:equal-values}
If $\vec{\sigma}$ is a Nash equilibrium of $\Gamma_{CT}$ such that $u_s(\vec{\sigma}) \ge \frac{1}{2}$ or $u_s(\vec{\sigma}) \ge \frac{1}{2}$, then $p_{i,i}(\vec{\sigma}) = p_{i, i+1}(\vec{\sigma}) = \frac{1}{2}$ and $p_{i,i-1}(\vec{\sigma}) = 0$ for all $i \in \{0,1,2\}$.
\end{proposition}

\begin{proof}
Suppose that there exists a Nash equilibrium $\vec{\sigma}$ such that $u_s(\vec{\sigma}) \ge \frac{1}{2}$ or $u_s(\vec{\sigma}) \ge \frac{1}{2}$. Let $h_r$ be a local history of the receiver in which the receiver plays action $1$. By Lemma~\ref{lemma-eq-2}, the probability that the sender's type is $t_2$ conditioned to the fact that the receiver's history is $h_r$ is $0$. Moreover, since $\vec{\sigma}$ is a Nash equilibrium, it means that the probability that the sender's type is $t_0$ given $h_r$ is greater than the probability that the sender's type is $t_1$. Since $q$ is the uniform distribution, this implies that $$p_{0,1}(\vec{\sigma}) \ge p_{1,1}(\vec{\sigma}).$$
Analogously, we have that $$
\begin{array}{c}
p_{1,2}(\vec{\sigma}) \ge p_{2,2}(\vec{\sigma})\\
p_{2,0}(\vec{\sigma}) \ge p_{0,0}(\vec{\sigma}).
\end{array}$$

However, by Lemma~\ref{lemma-eq-1} we have that $$p_{0,1}(\vec{\sigma}) + p_{1,2}(\vec{\sigma}) + p_{2,0}(\vec{\sigma}) = p_{1,1}(\vec{\sigma}) + p_{2,2}(\vec{\sigma}) + p_{0,0}(\vec{\sigma}),$$
and therefore the all inequalities above must actually be equal, which means that 
$$
\begin{array}{c}
p_{1,2}(\vec{\sigma}) = p_{2,2}(\vec{\sigma})\\
p_{2,0}(\vec{\sigma}) = p_{0,0}(\vec{\sigma})\\
p_{0,1}(\vec{\sigma}) = p_{1,1}(\vec{\sigma}).\\
\end{array}$$

Note that, by Lemma~\ref{lemma-eq-2}, $p_{i,i}(\vec{\sigma}) + p_{i, i+1}(\vec{\sigma}) = 1$ for all $i \in \{0,1,2\}$, and $p_{i, i-1}(\vec{\sigma}) = 0$. This, in addition to the previous equations, implies that $p_{i,i}(\vec{\sigma}) = p_{i, i+1}(\vec{\sigma}) = \frac{1}{2}$ for all $i \in \{0,1,2\}$.
\end{proof}

Proposition~\ref{prop:equal-values} allows the following refinement:

\begin{proposition}\label{prop:probability-05}
Let $\vec{\sigma}$ is a Nash equilibrium of $\Gamma_{CT}$ such that $u_s(\vec{\sigma}) \ge \frac{1}{2}$ or $u_s(\vec{\sigma}) \ge \frac{1}{2}$. Given a local history $h_r$ of the receiver when playing $\vec{\sigma}$, let $p_j(h_r)$ denote the probability that the sender's type is $t_j$ conditioned that the receiver's local history is $h_r$ and both players are following $\vec{\sigma}$. If $h_r$ is a local history where the receiver played action $i$, then $p_i(h_r) = p_{i-1}(h_r) = \frac{1}{2}$, and $p_{i+1}(h_r) = 0$.
\end{proposition}

\begin{proof}
An analogous argument to the one in the proof of Proposition~\ref{prop:equal-values} shows that, given a local history of the receiver in which the receiver plays action $i$, $p_{i+1}(h_r) = 0$ and $p_{i-1}(h_r) \ge p_i(h_r)$. By Proposition~\ref{prop:equal-values}, equality must hold since otherwise it would contradict the fact that $p_{i-1, i}(h_r) = p_{i,i}(h_r)$.
\end{proof}

Proposition~\ref{prop:probability-05} implies that, if $\vec{\sigma}$ is a Nash equilibrium of $\Gamma_{CT}$ such that $u_s(\vec{\sigma}) \ge \frac{1}{2}$ or $u_s(\vec{\sigma}) \ge \frac{1}{2}$, whenever the receiver has to play an action, there are only three possibilities given the information contained in its local history. Either (a) the chances of the sender's type being $0$ and $1$ are $1/2$ each, (b) the chances of the sender's type being $1$ and $2$ are $1/2$ each, or (c) the chances of the sender's type being $2$ and $0$ are $1/2$ each. Moreover, it also implies that in each of these scenarios, the receiver always plays actions $1$, $2$ and $0$ respectively. This means that, whenever the receiver realizes that the sender's type is \textbf{not} a given value (which, according to Proposition~\ref{prop:probability-05}, must always happen eventually), then it uniquely determines the action to play.

Consider the following strategy $\sigma'_s$ for the sender: it first chooses a type $t \in \{t_0, t_1, t_2\}$ uniformly at random and plays $\sigma_s$ pretending it has type $t$. However, if at some history $(h_s, h_r)$ the sender would send a sequence of messages $\vec{m}$ that make the receiver believe that its type cannot be $t_i$ (i.e., that $p_i(h_r+\vec{m}) = 0$, where $h_r + \vec{m}$ is the local history of the receiver after receiving $\vec{m}$), the sender does the following. If the sender's true type is $t_{i-1}$, the sender sends $\vec{m}$. Otherwise, from that communication round on, the sender sends the necessary messages $\vec{m'}$ that maximize $p_i(h_r + \vec{m}')$. By definition, if the sender plays this way, $p_i(h_r)$ cannot decrease from that point on, which means that the receiver will eventually play actions $i$ or $i+1$. By construction, the probability that the sender has type $t_{i-1}$ conditioned to the fact that the receiver has history $h_r$ is $\frac{1}{3}$ (note that the sender randomizes its ``fake'' type). In this case, by Proposition~\ref{prop:probability-05}, it is guaranteed that the sender gets $1$ utility by sending $\vec{m}$. In the remaining cases, the receiver ends up playing either $i$ or $i+1$ and, regardless of what the receiver plays, it is equally likely that the sender's true type is $t_i$ or $t_{i+1}$ since the sender's strategy does not make the distinction. Thus, the sender's expected utility is $$u_s((\sigma'_s, \sigma_r)) = \frac{1}{3} + \frac{2}{3} \cdot \frac{1}{2} = \frac{2}{3}.$$

Since $\vec{\sigma}$ is a Nash equilibrium, $u_s(\vec{\sigma}) \ge u_s((\sigma'_s, \sigma_r)) = \frac{2}{3}$ and thus, by Lemma~\ref{lemma-eq-2}, $u_r(\vec{\sigma}) \ge \frac{2}{3}$.
This would imply that $u_s(\vec{\sigma}) + u_r(\vec{\sigma}) > 1$, but this contradicts the fact that $u_r(t_i, j) + u_s(t_i, j) \le 1$ for all $i,j \in \{0,1,2\}$. This proves Proposition~\ref{prop:example2}.

\section{Notes on Information Aggregation}\label{sec:information-aggregation}

So far, all the results of this paper considered information transmission games, where are games in which a single sender has full knowledge of the world's state and an oblivious receiver must play an action that affects the utilities of both. However, we can extend some of the analysis to information aggregation games (i.e., games with several senders and a single receiver) as a consequence of the following result by Abraham, Dolev, Gonen and Halpern (ADGH from now on).

\begin{theorem}[\cite{ADGH06}]\label{thm:adgh}
    Let $\Gamma = (P, T, Q, A, U)$ be a Bayesian game with $|P| \ge 4$ and let $\vec{\sigma} + \sigma_d$ be a Nash equilibrium of $\Gamma_d$. Then, there exists a Nash equilibrium $\vec{\sigma}_{CT}$ of $\Gamma_{CT}$ such that $\mu_{\vec{\sigma} + \sigma_d} \equiv \mu_{\vec{\sigma}_{ACT}}$.
\end{theorem}

This result implies that, if there are at least four players, all communication with a mediator can be simulated with long cheap talk. This implies the following Corollary.

\begin{corollary}
    Let $\Gamma = (P, T, Q, A, U)$ be an information aggregation game with at least three senders and let $w$ be any positive welfare function for $\Gamma$. Then, $$M^w(\Gamma) = 1.$$
\end{corollary}

A second result from ADGH shows that the bound on the number of players in Theorem~\ref{thm:adgh} can be relaxed to $|P| \ge 3$ if $\Gamma$ has a \emph{punishment strategy} with respect to $\vec{\sigma}$. Intuitively, a punishment strategy w.r.t. $\vec{\sigma}$ is a strategy profile that, if played by everyone except one player, all players receive strictly less utility than the one they'd receive if everyone played $\vec{\sigma}$. Their result formally states the following.

\begin{theorem}[\cite{ADGH06}]\label{thm:adgh2}
    Let $\Gamma = (P, T, Q, A, U)$ is a Bayesian game with $|P| \ge 3$ and let $\vec{\sigma} + \sigma_d$ be a Nash equilibrium of $\Gamma_d$. Then, if there exists a punishment strategy with respect to $\vec{\sigma} + \sigma_d$, there exists a Nash equilibrium $\vec{\sigma}_{CT}$ of $\Gamma_{CT}$ such that $\mu_{\vec{\sigma} + \sigma_d} \equiv \mu_{\vec{\sigma}_{ACT}}$.
\end{theorem}

This implies that, in information aggregation games with two or more senders, a mediator can never Pareto-improve the utilities of all the senders and that of the receiver simultaneously with respect to long cheap talk:

\begin{corollary}
    Let $\Gamma = (P, T, Q, A, U)$ be an information aggregation game with two senders $s_1$ and $s_2$. Then, for every Nash equilibrium $\vec{\sigma} + \sigma_d$ of $\Gamma_d$ there exists a Nash equilibrium $\vec{\tau}$ of $\Gamma_{CT}$ such that $$u_{s_1}(\vec{\tau}) \ge u_{s_1}(\vec{\sigma} + \sigma_d) \quad \mbox{or} \quad u_{s_2}(\vec{\tau}) \ge u_{s_2}(\vec{\sigma} + \sigma_d) \quad \mbox{or} \quad u_{s_3}(\vec{\tau}) \ge u_{s_3}(\vec{\sigma} + \sigma_d)$$
\end{corollary}

\begin{proof}
Consider any Nash equilibrium $\vec{\sigma} + \sigma_d$ in $\Gamma_d$ and a Nash equilibrium $\vec{\tau}$ in $\Gamma_{CT}$. If $\vec{\sigma} + \sigma_d$ Pareto-improve $\vec{\tau}$ we'd have that $\vec{\tau}$ is a punishment strategy with respect to $\vec{\sigma} + \sigma_d$. By Theorem~\ref{thm:adgh2}, this would imply that $\vec{\sigma} + \sigma_d$ can be implemented with long cheap talk.
\end{proof}

\section{Conclusion and Further Work}

In this paper, we introduced a framework that allowed us to analyze the effects of mediated communication versus long cheap talk in Bayesian games. We showed that in information transmission games with binary actions, a mediator cannot improve the utility of the sender, the utility of the receiver, or any positive combination of these. However, we also showed that, in some cases, a mediator can greatly improve the utility of a third party that is not directly involved, even if the game has binary actions. If the game has more than two actions, a mediator can improve both the utility of the sender and the utility of the receiver at the same time.

For future work, it is still an open problem to find a precise characterization of the games in which a mediator provides an improvement over long cheap talk for a given welfare function, or when the mediator can Pareto-improve the utilities of the sender and the receiver. 

\bibliographystyle{plain}
\bibliography{bibfile}

\end{document}